\documentclass[11pt]{article}
\usepackage[T1]{fontenc}

\usepackage{graphicx}

\usepackage{soul}
\usepackage[a4paper,margin=1in]{geometry}
\usepackage{url}
\usepackage[hidelinks]{hyperref}
\usepackage[utf8]{inputenc}
\usepackage[small]{caption}
\usepackage{amsmath,amsthm}
\usepackage{algorithm}
\usepackage{algorithmic}
\usepackage[switch]{lineno}
\usepackage{multicol,multirow}

\usepackage{amstext,amssymb,amsfonts}
\usepackage{physics}
\usepackage{comment}
\usepackage{authblk}

\usepackage{pgfplots} \pgfplotsset{compat=1.15} \usepackage{mathrsfs} \usetikzlibrary{arrows} 

\newtheorem{theorem}{Theorem}
\newtheorem{definition}{Definition}
\newtheorem{clm}{Claim}

\newtheorem{corollary}{Corollary}

\definecolor{xdxdff}{rgb}{0.5,0.5,1}
\usepackage{enumitem}

\begin{document}
\title{Nearly Tight Bounds on Approximate Equilibria in Spatial Competition on the Line}
%
%
\author[1]{Umang Bhaskar\thanks{umang@tifr.res.in}}
\author[1]{Soumyajit Pyne\thanks{soumyajit.pyne@tifr.res.in}}
\affil[1]{Tata Institute of Fundamental Research, Mumbai}
\maketitle              

\begin{abstract}

    In Hotelling's model of spatial competition, a unit mass of voters is distributed in the interval $[0,1]$ (with their location corresponding to their political persuasion), and each of $m$ candidates selects as a strategy his distinct position in this interval. Each voter votes for the nearest candidate, and candidates choose their strategy to maximize their votes. It is known that if there are more than two candidates, equilibria may not exist in this model. It was unknown, however, how close to an equilibrium one could get. Our work studies approximate equilibria in this model, where a strategy profile is an (additive) $\epsilon$-equilibria if no candidate can increase their votes by $\epsilon$, and provides tight or nearly-tight bounds on the approximation $\epsilon$ achievable.

    We show that for 3 candidates, for any distribution of the voters, $\epsilon \ge 1/12$. Thus, somewhat surprisingly, for any distribution of the voters and any strategy profile of the candidates, at least $1/12$th of the total votes is always left ``on the table.'' Extending this, we show that in the worst case, there exist voter distributions for which $\epsilon \ge 1/6$, and this is tight: one can always compute a $1/6$-approximate equilibria. We then study the general case of $m$ candidates, and show that as $m$ grows large, we get closer to an exact equilibrium: one can always obtain an $1/(m+1)$-approximate equilibria in polynomial time. We show this bound is asymptotically tight, by giving voter distributions for which $\epsilon \ge 1/(m+3)$.
\end{abstract}

\section{Introduction}

In elections, strategic positioning by candidates is a common phenomenon. Candidates try to estimate where voters lie on a particular issue, such as through polls, past experience, or media reports, and then adopt positions accordingly. Depending on the opinion of the voters, candidates may position themselves to appear more conservative, liberal, or moderate, often even contradicting their earlier or stated positions.

Hotelling's seminal model, introduced in 1929, has been very influential in studying such strategic positioning by candidates~\cite{Hotelling29}. The initial model by Hotelling studied two merchants selling an identical good that sought to attract customers distributed in the interval $[0,1]$. Each merchant chose both a strategic location and a price for the identical good, and each customer then chose to purchase from the merchant which minimized the sum of distance and price paid. The primary message of Hotelling's work was that, in such competition in duopolies, competition leads to \emph{minimum differentiation}, where at equilibrium the two competitors end up being very close to each other.

Downs adopted Hotelling's basic model to examine electoral competition, even in scenarios involving more than two candidates~\cite{Downs57}. Since then, the model, and variations of it, have been very influential in political economics. This is then a spatial competition among the candidates, where each candidate seeks to maximize the number of votes they receive, in the presence of the other candidates. With more than three candidates, however, there may not be an equilibrium (and in fact for three candidates there is never an equilibrium!)~\cite{EatonL75,Osborne}. 

Despite its shortcomings, equilibria is a fundamental and very appealing solution. However, exact equilibrium may be too demanding a notion and the non-existence of equilibria is not informative if one could get very close to an equilibrium. On the other hand, if any strategy profile is far from an equilibrium, then the game is highly unstable, and equilibrium may not be a suitable model for the behaviour of agents. In this work, we hence focus on the question of \emph{how close can one get to an equilibrium} or \emph{how stable the instance is} within the basic model of spatial competition motivated by Hotelling's work. We call a strategy profile an (additive) $\epsilon$-equilibrium if no candidate can increase their votes by $\epsilon$, and provide tight or nearly-tight bounds on the approximation $\epsilon$ achievable. Despite the vast literature on this model and extensions to it, we believe work on quantifying the instability in the absence of equilibria deserves more attention.

We note that in many practical applications where exact solutions do not exist or are difficult to obtain, additive approximate solutions are commonly used. An example of this from social choice theory is the use of EF1 (envy-free up to one good) allocations for indivisible goods, which is an additive approximation~\cite{UFMNW}. For mixed Nash equilibria in bimatrix (i.e., two-person normal-form) games, additive approximations are more commonly studied. In bimatrix games, for constant $\epsilon$, quasi-polynomial time algorithms giving an additive $\epsilon$-approximation to Nash equilibria were known prior to the seminal PPAD-completeness result~\cite{ChenD06,LGSS}. The PPAD-completeness result was immediately followed by the hardness of computing an additive $\epsilon$-approximation, for polynomially small $\epsilon$~\cite{CNEASC}.

\subsection{Related Work}

The paper closest to ours in studying spatial competition with multiple candidates is by Eaton and Lipsey~\cite{EatonL75}. Their objective is to study conditions in which the minimum differentiation shown by Hotelling holds. Among other results, they describe necessary and sufficient conditions for the existence of equilibria and show that if there are three candidates, then an equilibrium does not exist. For $m \ge 4$ candidates, they show that if $m$ exceeds twice the number of local maxima in the distribution of voters, there is no equilibrium. Hence in particular for a unimodal distribution of voters (such as a truncated Gaussian distribution), there does not exist an equilibrium for any number of candidates. They also study equilibria when the market is a circle, rather than the interval $[0,1]$, as well as the multidimensional setting when the customers are distributed in a disc. The existence of equilibrium is also studied by Fournier and Scarsini~\cite{FournierS19}, albeit in a more general model where the voters are present in a graph. They show that for \emph{uniform} distribution of voters, if the number of candidates is large enough, then an equilibrium exists, and they also study the inefficiency (i.e., the Price of Anarchy and Price of Stability) of equilibria.

Several variants of Hotelling's model are also studied. This literature is quite vast, and hence we give some pointers rather than a comprehensive survey. One variant studies the case of competition between parties, each of which can choose one of multiple candidates located on the line~\cite{HarrensteinLST21,DeligkasEG22}. The goal of each party is to maximize their votes, given the candidates chosen by other parties. Other studied variations include models where candidates can enter and exit the election~\cite{FedersenSW90,SenguptaS08}, when candidates strategize to win the election, rather than maximizing their votes~\cite{ChisikL06}, when some candidates have fixed positions and do not strategize~\cite{JonesSF22}, and when voters do not vote at all when there is no candidate sufficiently close to their position~\cite{FeldmanFO16,Shen,JonesSF22}. For a more detailed survey of these results, we refer to~\cite{EiseltMD19,EnelowH90}.

\subsection{Our Contribution}
We first focus on the case of approximate equilibria with 3 candidates. This is the simplest case beyond 2 candidates and is interesting because, as shown earlier, there is no equilibrium irrespective of the distribution of voters~\cite{EatonL75}. We first show tight bounds on the distance of any strategy profile from equilibria, for any voter distribution. We show that for \emph{any} distribution, one cannot obtain better than a $1/12$-approximate equilibrium. Thus for any distribution of the voters, and for any location of the three candidates, some candidate can increase his votes by at least $1/12$ (or $8.5\%$ of the total vote) by choosing a different position. Given that modern elections often hinge on a small percentage of the total vote, our findings suggest that any scenario involving three candidates is highly \emph{unstable}. We also show that this is the worst possible: there exists a distribution where one can in fact obtain a $1/12$-approximate equilibrium.

The bound of $1/12$ is true for all distributions; an immediate question then is about the worst-case voter distribution. We next show that there exists a voter distribution where one cannot obtain better than a $1/6$-approximate equilibrium. 
We also show that the approximation $1/6$ is tight, showing that for any distribution, one can obtain a $1/6$-approximate equilibrium.
\bgroup
\def\arraystretch{1.5}
\setlength\tabcolsep{0.1in}
\begin{table}
\begin{center}
\begin{tabular}{|l|c|c|c|c|}
\hline
& \multicolumn{2}{|c|}{$m=3$ candidates} & \multicolumn{2}{|c|}{$m \ge 4$ candidates} \\ \hline
& $\exists$ distrib. & $\forall$ distribs. & $\exists$ distrib. & $\forall$ distribs. \\ \hline 
Lower bound & $\epsilon \ge 1/6$ & $\epsilon \ge 1/12$ & $\epsilon \ge 1/(m+3)$ & 0 \\ \hline 
Upper bound & $\epsilon \le 1/12$ & $\epsilon \le 1/6$ & 0 & $\epsilon \le 1/(m+1)$ \\ \hline 
\end{tabular}
\end{center}
\caption{Our results on the additive approximation achievable. The ``$\exists$ distrib." columns show results for worst-case (for lower bounds) / best-case (for upper bounds) distributions, and the ``$\forall$ distribs." columns show results over all distributions. Our bounds are tight for 3 candidates, and asymptotically tight for 4 or more candidates.\label{table:results1}}
\end{table}
\egroup
Moving beyond 3 candidates, we then consider approximate equilibria for $m$ candidates. It is known that there are distributions for which equilibria exist for $m \ge 4$ candidates~\cite{EatonL75}, hence we are concerned with worst-case distributions. Here, we again show nearly tight bounds. Specifically, we show that for any distribution, one can obtain a $\frac{1}{m+1}$-approximate equilibrium, and there exist distributions for which one cannot obtain better than a $\frac{1}{m+3}$-approximate equilibrium. Thus our results present a quantitative perspective on previous results on the nonexistence of equilibria.  The upper bounds of our model indicate that as the number of candidates increases, the model becomes increasingly \emph{stable}, regardless of the distribution of voters. Overall our results present tight or nearly tight bounds on approximate equilibria. Our upper bounds are via polynomial time algorithms that have access to an oracle that supports two functions: (i) given $x \in [0,1]$, returns $F(x)$, the total voters until the location $x$, and (ii) given $x, v \in [0,1]$, returns a location $y = \text{Cut}(x,v) \ge x$ so that there are exactly $v$ voters in the interval $[x,y]$.\footnote{Note that the second query --- \text{Cut}($x$,$v$) --- can be simulated to arbitrary precision by running a binary search using just the first query.} These are the same as Eval and Cut queries used in the Robertson-Webb query model for fair cake division~\cite{RobertsonW98}. It's worth mentioning that the simplicity of our algorithm in determining the upper bound of $\frac{1}{m+1}$ indicates that a simpler access method suffices. Specifically, an oracle that provides the positions of the $m$th quantiles of the voter distribution for any $m$ suffices. This is similar to the oracle utilized by the GLIME mechanism~\cite{GLIME}.

Lastly, we note that our results are for the case where candidates may occupy locations near each other, but cannot occupy the same location in the $[0,1]$ interval. A natural question then is if the results change if multiple candidates are allowed to occupy the same location. This could be interpreted as voters being unable to perfectly distinguish between candidates that are very near to each other. We show this latter model may give different results. Specifically, for scenarios involving three candidates, we show that regardless of the distribution of voters, a $1/7$-approximate equilibrium can be achieved, thus offering greater stability compared to the previous model, for which there exist distributions for which $\epsilon$ must be at least $1/6$. Furthermore, this bound is tight, as there exists a distribution where obtaining a better approximation than $1/7$ is impossible. We leave further study in this model for future work.

\bgroup
\def\arraystretch{1.5}
\setlength\tabcolsep{0.1in}
\begin{table}
\begin{center}
\begin{tabular}{|l|c|c|}
\hline
& \multicolumn{2}{|c|}{$m=3$ candidates} \\ \hline
& $\exists$ distrib. & $\forall$ distribs. \\ \hline 
Lower bound & $\epsilon \ge 1/7$ & ?  \\ \hline 
Upper bound & ? & $\epsilon \le 1/7$  \\ \hline 
\end{tabular}
\end{center}
\caption{Our results on the additive approximation achievable for the variant where multiple candidates can occupy the same location. Question marks indicate open questions. Note the better approximation achievable, as compared to Table~\ref{table:results1}.\label{table:results2}}
\end{table}
\egroup

\section{Preliminaries and Notation}

A unit mass of voters is distributed in the interval $[0,1]$ according to a density function $f:[0,1]\rightarrow\mathbb{R^+}\cup \{0\}$ such that $f$ is integrable, bounded so that $f(z) \le M$ for some finite $M$, and $\int_0^1 f(z) \, dz = 1$. Let $F(y) = \int_0^y f(z) \, dz$ be the integral. Since $f(z) \le M$, for any interval of length $\delta$, the total voters in the interval is at most $M \delta$.

There are $m$ candidates. Each candidate $i$ chooses a real number $x_i$ from $[0,1]$ as their strategy, hence $X=(x_1,x_2,...,x_m)$ is a strategy profile. We restrict the candidates to occupy distinct positions in the interval, so that $\min_{i,j} |x_i - x_j|$ $\ge \delta$ for some small $\delta$. It is helpful to think of $\delta$ as approaching zero. In particular, we will assume $M \delta < 10^{-3}$. For a strategy profile $X$, we will use $X_{-i}$ to denote the position of candidates other than $i$.

Given a strategy profile $X$, assume without loss of generality that $x_1 \le x_2 \le \ldots \le x_m$.\footnote{We will make this assumption whenever possible to avoid cumbersome notation.} Then for a candidate $1<i<m$, the voters located between $(x_{i-1} + x_i)/2$ on the left and $(x_i + x_{i-1})/2$ on the right are nearer candidate $i$ than any other candidate $j$, and will hence vote for $i$. We then define the \emph{left} and \emph{right votes} (or \emph{utility}) of candidate $i$ as 

\[
 U_i^L(X) = \int\limits_{\frac{x_{i-1} + x_i}{2}}^{x_i} f(z) \, dz\, , ~\text{and}~ U_i^R(X) = \int\limits_{x_i}^{\frac{x_{i} + x_{i+1}}{2}} f(z) \, dz \, .
\]

\noindent We need to define the utilities separately for candidate 1 and $m$.
\[
 U_1^L(X) = \int\limits_{0}^{x_1} f(z) \, dz\, = \, F(x_1) , ~\text{and}~ U_1^R(X) = \int\limits_{x_1}^{\frac{x_{1} + x_{2}}{2}} f(z) \, dz \, .
\]
\[
 U_m^L(X) = \int\limits_{\frac{x_{m-1} + x_m}{2}}^{x_m} f(z) \, dz\, , ~\text{and}~ U_m^R(X) = \int\limits_{x_m}^{1} f(z) \, dz\, = \, 1-F(x_m) .
\]

\noindent The total votes or utility of candidate $i\in[m]$, is $U_i(X) = U_i^L(X) + U_i^R(X)$. Note that the total votes over all candidates is always 1.

A strategy profile $X$ is an equilibrium if for all candidates $i$, and all locations $x'$ that satisfy $|x' - x_j| \ge \delta$ for all $j \neq i$, $U_i(X) \ge U_i(x', X_{-i})$. It is known that for two candidates, the positions $x_1 = \mu - \delta/2$, $x_2 = \mu + \delta/2$ are an equilibrium, where $\mu$ is the median (i.e., $F(\mu)=1/2$). However, for three candidates, no equilibrium exists. Our goal in this paper is to study approximate equilibria.

\begin{definition}[$\epsilon$-equilibrium]
    Given an $\epsilon \ge 0$, $X = (x_1,x_2,...,x_m)$ is an $\epsilon$-equilibrium if for any candidate $i\in[m]$, any location $x_i'\in[0,1]$, and $\forall j\neq i$, $|x_i'-x_j|\geq\delta$  \[\lim_{\delta\to0} (U_{i}(X')-U_{i}(X))\leq\epsilon\] where $X'=(x_i',X_{-i})$.
\end{definition}
Note that $X$ is a \textit{pure Nash equilibrium} if the condition is satisfied with $\epsilon=0$.

For our algorithms to compute equilibria, in order to access the voter density function $f$, we will assume we have access to an oracle that supports the following queries:
\begin{itemize}
    \item $F(z)$: Returns $F(z)$, the total voters in the interval $[0,z]$.
    \item Cut($z,v$): Given a location $z \in [0,1]$ and a value $v \in [0,1]$, returns a location $y \ge z$ so that $F(y) - F(z) = v$, or returns $1$ if there is no such $y$. If there are multiple such locations, return one arbitrarily but consistently.
\end{itemize}

\section{Approximate Equilibria for Three Candidates}

In this section, we are interested in three candidates, i.e., $m=3$. Each candidate $i\in\{1,2,3\}$ chooses a real number $x_i$ from $[0,1]$ as their strategy such that $x_1<x_2<x_3$. We begin by showing a lower bound of $1/12$ on $\epsilon$.

\begin{theorem}\label{lb112}
    If $X = (x_1,x_2,x_3)$ is an $\epsilon$-equilibrium, then $\epsilon\geq\frac{1}{12} - M \delta$.
\end{theorem}

The proof of the theorem is based on three claims, the first two of which bound the left and right votes for each candidate, and the last of which bounds the votes between $x_1$ and $x_3$.

\begin{clm}\label{mid}
$U_2^L(X)$ and $U_2^R(X)$ are both at most $\epsilon + M \delta$.
\end{clm}

\begin{proof}
For a contradiction, assume $U_2^L(X) > \epsilon + M\delta$. Now let $x_1' = x_2 - \delta$. Then candidate 1 will get its earlier votes and the left votes of candidate 2 minus at most $M \delta$, resulting in an increase of more than $\epsilon$, which is a contradiction. That is, if $X' = (x_1', x_2, x_3)$, then 

\begin{align*}
    U_1^L(X') &=  F(x_2 - \delta) \\
        & = F((x_1+x_2)/2) + F(x_2 - \delta) - F((x_1+x_2)/2) \\
        & = U_1(X) + F(x_2 - \delta) - F((x_1+x_2)/2) \\
        & \ge U_1(X) + F(x_2) - M \delta - F((x_1+x_2)/2) \\
        & = U_1(X) + U_2^L(X) - M \delta \\
        & > U_1(X) + \epsilon + M \delta - M \delta  = U_1(X)+\epsilon\, .
\end{align*}

\noindent Hence shifting to $x_2 - \delta$ increases candidate 1's votes by more than $\epsilon$, which is a contradiction. By a similar argument, we get that $U_2^R(X) \le \epsilon + M\delta$ as well. 
\end{proof}

\begin{clm}\label{extreme}
$U_1^L(X)$ and $U_3^R(X)$ are both at most $ 3(\epsilon + M \delta)$.
\end{clm}

\begin{proof}
From Claim~\ref{mid}, candidate 2 has utility at most $2(\epsilon + M \delta)$. Now if $U_1^L(X) > 3(\epsilon + M \delta)$, let $x_2' = x_1 - \delta$, and consider the strategy profile $X' = (x_1, x_2', x_3)$. Then candidate 2's vote is at least $F(x_1 - \delta) 
  \ge U_1^L(X) - M \delta$, and since $U_1^L(X) > 3(\epsilon + M \delta)$, this is greater than $3\epsilon + 2 M \delta$. Hence, $X$ cannot be an $\epsilon$-equilibrium, giving a contradiction. Similarly, we obtain $U_3^R(X) \leq 3(\epsilon+M \delta)$.
\end{proof}

\begin{clm}\label{midtot}
$F((x_1+x_3)/2) - F(x_1)$ and $F(x_3) - F(x_1+x_3)/2)$ are both at most $3(\epsilon + M \delta)$.
\end{clm}

The proof is similar to Claim~\ref{extreme}. We show that if the claim does not hold, then for candidate 2, shifting to either $x_1 + \delta$ or $x_3- \delta$ increases his votes by more than $\epsilon$.

We now complete the proof of the theorem. Figure~\ref{fig:firstproof} shows the bounds on the votes from Claim~\ref{mid} and Claim~\ref{extreme}. 

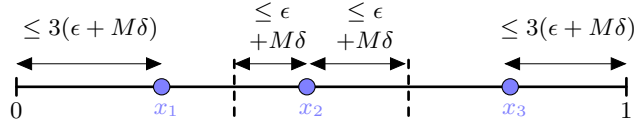
\begin{figure}[H]
    \centering
    \begin{tikzpicture}[line cap=round,line join=round,>=triangle 45,x=1cm,y=1cm,scale=0.382]
    \clip(-1,-3) rectangle (22,4);
\draw [dashed,line width=1pt] (13.5,-1) -- (13.5,1);
\draw [dashed,line width=1pt] (7.5,-1) -- (7.5,1);
\draw [line width=1pt] (21,-0.3) -- (21,0.3);
\draw [line width=1pt] (0,-0.3) -- (0,0.3);
\draw [line width=1pt] (0,0) -- (21,0);
\draw [line width=0.3pt,<->] (0,0.8) -- (5,0.8);
\draw [line width=0.3pt,<->] (16.8,0.8) -- (21,0.8);
\draw [line width=0.3pt,<->] (7.5,0.8) -- (10,0.8);
\draw [line width=0.3pt,<->] (10.1,0.8) -- (13.4,0.8);
\begin{footnotesize}
\draw [fill=xdxdff] (5,0) circle (8pt);
\draw [fill=xdxdff] (17,0) circle (8pt);
\draw [fill=xdxdff] (10,0) circle (8pt);
\draw[color=xdxdff] (10.16,-0.8) node {$x_2$};
\draw[color=xdxdff] (5.16,-0.8) node {$x_1$};
\draw[color=xdxdff] (17.16,-0.8) node {$x_3$};
\draw[color=black] (0,-0.8) node {$0$};
\draw[color=black] (21,-0.8) node {$1$};
\draw[color=black] (2.55,2) node {$\le 3(\epsilon + M \delta)$};
\draw[color=black] (8.8,2.6) node {$\le \epsilon $};
\draw[color=black] (9,1.6) node {$+ M\delta$};
\draw[color=black] (12,2.6) node {$\le \epsilon $};
\draw[color=black] (12,1.6) node {$+M \delta$};
\draw[color=black] (19,2) node {$\le 3(\epsilon + M \delta)$};
\end{footnotesize}
\end{tikzpicture}
    \caption{Figure showing bounds on the votes, as shown by Claim~\ref{mid} and Claim~\ref{extreme}. Candidates are shown by blue circles. The dashed lines are the mid-points of $x_1,x_2$ and $x_2, x_3$.}
    \label{fig:firstproof}
\end{figure}

\begin{proof}[of Theorem~\ref{lb112}]
        We can write the total votes as 
        
        \begin{align}
            1 & = F(x_1) + \left( F((x_1+x_3)/2) - F(x_1) \right) \nonumber \\
                & \qquad + \left(F(x_3) - F((x_1+x_3)/2)\right) + \left(1 - F(x_3)\right) \label{eqn:sum} \, .
        \end{align}

        Now, from Claim~\ref{extreme} we get that $F(x_1) = U_1^L(X) \le 3(\epsilon + M \delta)$, and $1-F(x_3) = U_3^R(X) \le 3( \epsilon + M \delta)$. From Claim~\ref{midtot} we get that the remaining two terms are also at most $3(\epsilon + M \delta)$. Substituting in~\eqref{eqn:sum}, we get that $1 \le 12 (\epsilon + M \delta)$, or $\epsilon \ge \frac{1}{12} - M \delta$, as required.
\end{proof}

We now show that the bound in Theorem~\ref{lb112} is tight: we can't obtain a worse bound than $1/12$ that holds for all distributions. In particular, we exhibit a distribution for which a $1/12$-equilibrium is obtainable.

\begin{theorem}\label{112ex}
    There exists a distribution $f$ of voters for which there exists a $\frac{1}{12}$-equilibrium.
\end{theorem}

\begin{proof}
    Fix $\epsilon = 1/12$. We define the following symmetric distribution of voters:
         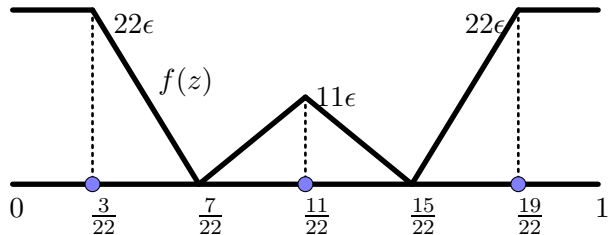
\begin{figure}[H]
         \centering
          \begin{tikzpicture}[line cap=round,line join=round,>=triangle 45,x=1cm,y=0.3cm,scale=0.35] \clip(-1,-8) rectangle (23,23); \draw [line width=2pt] (0,0)-- (22,0); \draw [line width=2pt] (0,22)-- (3,22); \draw [line width=2pt] (3,22)-- (7,0); \draw [line width=2pt] (7,0)-- (11,11); \draw [line width=2pt] (11,11)-- (15,0); \draw [line width=2pt] (15,0)-- (19,22); \draw [line width=2pt] (19,22)-- (22,22); \draw [dotted,line width=1pt] (19,22)-- (19,0); \draw [dotted,line width=1pt] (11,11)-- (11,0); \draw [dotted,line width=1pt] (3,22)-- (3,0); \draw (-0.5,-0.5) node[anchor=north west] {$0$}; \draw (2.5,-0.5) node[anchor=north west] {$\frac{3}{22}$}; \draw (3.4,22.8) node[anchor=north west] {$22\epsilon$}; \draw (16.6,22.8) node[anchor=north west] {$22\epsilon$}; \draw (11,13.5) node[anchor=north west] {$11\epsilon$}; \draw (6.5,-0.5) node[anchor=north west] {$\frac{7}{22}$}; \draw (10.5,-0.5) node[anchor=north west] {$\frac{11}{22}$}; \draw (14.5,-0.5) node[anchor=north west] {$\frac{15}{22}$}; \draw (18.5,-0.5) node[anchor=north west] {$\frac{19}{22}$}; \draw (21.5,-0.5) node[anchor=north west] {$1$}; \draw (6.5,13) node {$f(z)$}; \draw [fill=xdxdff] (3,0) circle (8pt);\draw [fill=xdxdff] (19,0) circle (8pt);\draw [fill=xdxdff] (11,0) circle (8pt);
          \end{tikzpicture}
         \caption{Distribution of voters and candidates for the lower bound of $1/12$ in Theorem \ref{112ex}. Candidates are indicated by blue circles.}
         \label{fig:sawtooth}
     \end{figure}

    \[ 
f(z) = 
     \begin{cases}
       22\epsilon &\quad 0\leq z<\frac{3}{22}\\
       121\epsilon\left(\frac{7}{22}-z\right) &\quad\frac{3}{22}\leq z<\frac{7}{22}\\
       \frac{121}{2}\epsilon\left(z-\frac{7}{22}\right) &\quad\frac{7}{22}\leq z<\frac{11}{22}\\
       \frac{121}{2}\epsilon\left(\frac{15}{22}-z\right) &\quad\frac{11}{22}\leq z<\frac{15}{22}\\
       121\epsilon\left(z-\frac{15}{22}\right) &\quad\frac{15}{22}\leq z<\frac{19}{22}\\
       22\epsilon &\quad\frac{19}{22}\leq z\leq1\\
     \end{cases}
     \]
     Consider the strategy profile $X=(x_1,x_2,x_3) = (\frac{3}{22},\frac{1}{2},\frac{19}{22})$. We will prove that $X$ is an $\epsilon$-equilibrium. For this profile, $U_{1}(X)= F(7/22) = 5 \epsilon$, $U_2(X) = F(15/22) - F(7/22) = 2 \epsilon$, and $U_3(X) = 1 - F(15/22) = 5 \epsilon$.

     For candidates 1 and 3, since $x_2 = 1/2$, they can at best deviate to obtain half the votes, and can hence increase their votes by at most $1/2 - 5\epsilon$ $=\epsilon$. For candidate 2, if he deviates to any $x' \in (x_1, x_3)$, he gets the votes in an interval of length $(x_3 - x_1)/2$, and it can be checked that the interval of this length with the maximum votes is centred at $x' = 7/22$ and $x' = 15/22$, each of which gets $3 \epsilon$ votes. Further if candidate 2 deviates to $x' \not \in [x_1, x_3]$, he gets at most $3 \epsilon$. In either case, no candidate can increase their votes by more than $\epsilon$, and hence this is an $\epsilon = 1/12$-equilibrium.
\end{proof}

Our previous lower bound holds irrespective of the voter distribution. We next consider worst-case voter distributions. How much further from an exact equilibrium are we pushed? Our next two theorems give tight bounds of $1/6$ on the worst-case approximation.

\begin{theorem}
    For any constant $\epsilon<\frac{1}{6}$, there exists a distribution $f$ of voters for which there does not exist an $\epsilon$-equilibrium with three candidates as $\delta \rightarrow 0$.
    \label{thm:three-candidates-lb}
\end{theorem}

This theorem is a special case of the more general result proved in Theorem~\ref{thm:general-lb}, where we show a lower bound of $1/(m+3)$ in the worst case if there are $m$ candidates. We thus defer the proof until later. Our next theorem however gives the required upper bound.

\begin{theorem}
Given any distribution $f$ of voters, a $(1/6 + M \delta)$-equilibrium can be found for three candidates in polynomial time.
\label{thm:ub-1/6}
\end{theorem}

\begin{proof}
    Let $x_1 = \text{Cut}(0,1/3)$ and $x_3 = \text{Cut}(0,2/3)$. Then $F(x_1) = F(x_3) - F(x_1) = 1 - F(x_3) = 1/3$. Keeping $x_1$ and $x_3$ fixed as the equilibrium locations for candidates 1 and 3, we will try and find a location for candidate 2 to obtain an equilibrium. Notice that either $F((x_1 + x_3)/2) - F(x_1) \ge 1/6$, or $F(x_3) - F((x_1+x_3)/2) \ge 1/6$. Without loss of generality, assume $F((x_1 + x_3)/2) - F(x_1) \ge 1/6$. Since $F$ is continuous, there must exist $x \in (x_1, x_3)$ so that $F((x+x_3)/2) - F(x) = 1/6$. Let $x_2$ be such a location $x$. If $|x_2 - x_1| < \delta$, let $x_2 = x_1 + \delta$.\footnote{Since $F(x_3) - F(x) \ge 1/6$, and we assume $M\delta < 10^{-3}$, clearly $x < x_3 - \delta$.}  We will now show that $X = (x_1, x_2, x_3)$ is in fact a $1/6$-equilibrium.

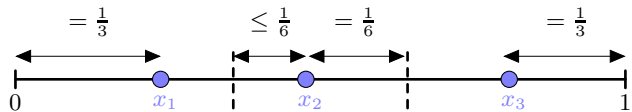
\begin{figure}[H]
    \centering
    \begin{tikzpicture}[line cap=round,line join=round,>=triangle 45,x=1cm,y=1cm,scale=0.382]
    \clip(-1,-3) rectangle (22,4);
\draw [dashed,line width=1pt] (13.5,-1) -- (13.5,1);
\draw [dashed,line width=1pt] (7.5,-1) -- (7.5,1);
\draw [line width=1pt] (21,-0.3) -- (21,0.3);
\draw [line width=1pt] (0,-0.3) -- (0,0.3);
\draw [line width=1pt] (0,0) -- (21,0);
\draw [line width=0.3pt,<->] (0,0.8) -- (5,0.8);
\draw [line width=0.3pt,<->] (16.8,0.8) -- (21,0.8);
\draw [line width=0.3pt,<->] (7.5,0.8) -- (10,0.8);
\draw [line width=0.3pt,<->] (10.1,0.8) -- (13.4,0.8);
\begin{footnotesize}
\draw [fill=xdxdff] (5,0) circle (8pt);
\draw [fill=xdxdff] (17,0) circle (8pt);
\draw [fill=xdxdff] (10,0) circle (8pt);
\draw[color=xdxdff] (10.16,-0.8) node {$x_2$};
\draw[color=xdxdff] (5.16,-0.8) node {$x_1$};
\draw[color=xdxdff] (17.16,-0.8) node {$x_3$};
\draw[color=black] (0,-0.8) node {$0$};
\draw[color=black] (21,-0.8) node {$1$};
\draw[color=black] (2.55,2) node {$=\frac{1}{3}$};
\draw[color=black] (8.8,2) node {$\le\frac{1}{6}$};
\draw[color=black] (11.7,2) node {$=\frac{1}{6}$};
\draw[color=black] (19,2) node {$=\frac{1}{3}$};
\end{footnotesize}
\end{tikzpicture}
    \caption{Figure showing the votes, as shown by Theorem~\ref{thm:ub-1/6}. Blue circles show candidates. The dashed lines are the mid-points of $x_1,x_2$ and $x_2, x_3$.}
    \label{fig:ub-1/6}
\end{figure}
    
    For candidate 2, $U_2^R(X) = F((x_2 + x_3)/2) - F(x_2) \ge 1/6 - M \delta$. Locations $x_1$ and $x_3$ equally partition the set of voters, hence the maximum votes that candidate 2 can get is $1/3$, and hence he cannot increase his votes by more than $1/6 + M \delta$.

    For candidate 1, $U_1^L(X) = 1/3$. Since $U_2^R(X) \ge 1/6 - M \delta$ and $U_3^R(X) = 1/3$, it follows that $F(x_2) \le 1/2 + M \delta$. Hence $x_2$ and $x_3$ partition the set of voters so that at most $1/2 + M \delta$ are to the left of $x_2$, at most $1/3$ between $x_2$ and $x_3$, and $1/3$ to the right of $x_3$. It follows that candidate 1 cannot increase his votes by more than $1/6 + M \delta$.

    For candidate 3, $U_1^L(X) = 1/3$. Again we note that $x_1$ and $x_2$ partition the set of voters so that $1/3$ voters are to the left of $x_1$, and at most $1/3$ between $x_1$ and $x_2$. Hence deviating to the left of $x_2$ cannot increase his votes. If he deviates to the right of $x_2$, he can at most gain the right votes of candidate 2 which is $U_2^R(X) \le 1/6$. Hence candidate 3 also cannot increase his votes by more than $1/6$, and hence $X$ is a $(1/6 + M \delta)$-equilibrium.
 \end{proof}
\section{Approximate Equilibria for $m$ Candidates}

How do the previous results change as we consider more candidates? As previously shown~\cite{Hotelling29}, for $m \geq 4$, there exist voter distributions that admit equilibria, hence in the best case $\epsilon = 0$. Hence we focus on the worst case and show that the approximation improves as $m$ increases.

\begin{theorem}
    Given a distribution $f$ of voters, a $\frac{1}{m+1}$-equilibrium can be found in polynomial time.
    \label{thm:general-ub}
\end{theorem}
\begin{proof}
    Our $\frac{1}{m+1}$-equilibrium is simply the strategy profile where the candidates form an equipartition of the voters. I.e., we choose $x_1 = \text{Cut}(0,\frac{1}{m+1})$, and $x_k = \text{Cut}(x_{k-1},\frac{1}{m+1})$ for $k=2, \ldots, m$. Note that then $1-F(x_m) = \frac{1}{m+1}$. We now show that $X = (x_1,x_2,...,x_m)$ is indeed a $\frac{1}{m+1}$-equilibrium. For ease of notation, let us also define $x_0 = 0$ and $x_{m+1} = 1$.

    For any candidate $k$, deviating to a location in another interval $x_k' \in (x_{l-1}, x_l)$ where $l \notin \{k,k+1\}$ cannot increase his votes by more than $\frac{1}{m+1}$, since by design $F(x_l) - F(x_{l-1}) = \frac{1}{m+1}$. Hence for any candidate $k$ we only need to consider deviations to $x_k' \in (x_{k-1}, x_{k+1})$. For this, assume $x_k' < x_k$. If candidate $k$'s right boundary $(x_k' + x_{k+1})/2$ is left of $x_k$, then the entire votes for $x_k'$ are in the interval $ (x_{k-1}, x_k)$, and in this case candidate $k$'s new utility is at most $\frac{1}{m+1}$. If the right boundary $(x_k' + x_{k+1})/2$ is to the right of $x_k$, then we get the following bound. Figure~\ref{fig:general-ub} depicts the candidate positions and their boundaries in this case.

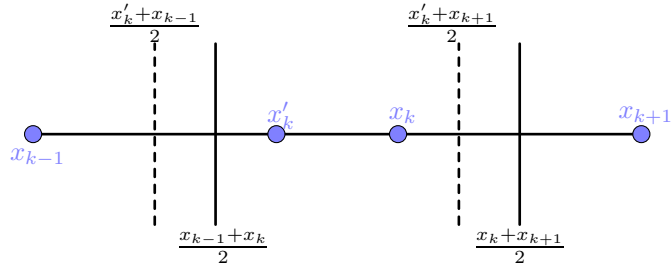
\begin{figure}[!htp]
    \centering
    \begin{tikzpicture}[line cap=round,line join=round,>=triangle 45,x=1cm,y=1cm,scale=0.4]
\clip(-1,-5) rectangle (21,5);
\draw [line width=1pt] (6,-3) -- (6,3);
\draw [dashed,line width=1pt] (14,-3) -- (14,3);
\draw [dashed,line width=1pt] (4,-3) -- (4,3);
\draw [line width=1pt] (16,-3) -- (16,3);
\draw [line width=1pt] (0,0) -- (20,0);
\begin{small}
\draw [fill=xdxdff] (0,0) circle (8pt);
\draw[color=xdxdff] (0.16,-0.78) node {$x_{k-1}$};
\draw [fill=xdxdff] (20,0) circle (8pt);
\draw[color=xdxdff] (20.16,0.62) node {$x_{k+1}$};
\draw [fill=xdxdff] (8,0) circle (8pt);
\draw[color=xdxdff] (8.16,0.62) node {$x_{k}'$};
\draw [fill=xdxdff] (12,0) circle (8pt);
\draw[color=xdxdff] (12.16,0.62) node {$x_k$};
\draw[color=black] (6.24,-3.7) node {$\frac{x_{k-1}+x_k}{2}$};
\draw[color=black] (13.76,3.7) node {$\frac{x_k'+x_{k+1}}{2}$};
\draw[color=black] (4,3.7) node {$\frac{x_k'+x_{k-1}}{2}$};
\draw[color=black] (16,-3.7) node {$\frac{x_k+x_{k+1}}{2}$};
\end{small}
\end{tikzpicture}
    \caption{Change in votes for candidate $k$, if he shifts from $x_k$ to $x_k'$.}
    \label{fig:general-ub}
\end{figure}

    \begin{align*}
        U_2(X') & = F((x_k' + x_{k+1})/2)  - F((x_k' + x_{k-1})/2) \\
            & = (F(x_k) - F((x_k' + x_{k-1})/2)  \\
            &\hspace{1.5cm}+ F((x_k' + x_{k+1})/2) - F(x_k))
            \end{align*}
    \noindent Since $x_{k-1} \le x_k' < x_k$, 
    \begin{align*}
        U_2(X')     & \le (F(x_k) - F(x_{k-1}) + \\
            &\hspace{1.5cm}F((x_k + x_{k+1})/2) - F(x_k)) \\
            & = \frac{1}{m+1} + U_2^R(X) \le \frac{1}{m+1} + U_2(X) \, .
    \end{align*}
    
    \noindent Thus, no candidate can increase their votes by more than $\frac{1}{m+1}$, and the strategy profile $X$ is an $\frac{1}{m+1}$-equilibrium. 
\end{proof}

We now define a particular rapidly decreasing function that will be very useful in defining voter density functions for showing lower bounds for approximate equilibria.

\begin{theorem}\label{magicfunction}
    Given $A > \gamma>0$ and $l \in [0,1]$, there exists an integrable and bounded function $g:[0,l]\rightarrow\mathbb{R^+}\cup\{0\}$ that satisfies (i) $\int_{0}^{l}g(z)\,dz = A$, and (ii) for all $y\in[0,l]$, $\int_{y/2}^{y}g(z)\,dz\leq\gamma$.
\end{theorem}
\begin{proof}
    Define $g$ as:
    \[   
g(z) = 
     \begin{cases}
     0 &\quad\text{for } z< le^{-\frac{A}{\gamma}} \, ,\\ 
       \frac{\gamma}{z} &\quad \text{for } le^{-\frac{A}{\gamma}}\leq z\leq l \, .
     \end{cases}\]
     Let $G(y)=\int_{0}^{y}g(z)\,dz$. For $y \le le^{-\frac{A}{\gamma}}$, $G(z) = 0$, and for $z\geq le^{-\frac{A}{\gamma}}$, $G(z) = A + \gamma\log\left(\frac{z}{l}\right)$.
     Notice that $G(l) = A$, and for any $y\in[0,l]$, $G(y) - G(y/2) = \gamma \ln 2$ $< \gamma$.
\end{proof}

The function derived from Theorem \ref{magicfunction} is discontinuous. However the function can be smoothed in the intuitive way so that it is continuous and differentiable (though the derivatives may be large), and all other required properties are maintained. We skip a formal description of the resulting function owing to the complexity.

\begin{theorem}
    Given $m$ candidates, there exists a distribution $f$ of voters such that, for any $\epsilon$-equilibrium, $\epsilon \ge \frac{1}{m+3}$ as $\delta \rightarrow 0$.
    \label{thm:general-lb}
\end{theorem}
\begin{proof}
    Fix $\gamma > 0$. We will eventually take the limit $\gamma \rightarrow 0$, to obtain the bound of $\frac{1}{m+3}$. We choose the distribution $f$ of the voters as per Theorem~\ref{magicfunction}, so that $f$ is bounded, integrable, $\int_0^1 f(z) dz = 1$, and for any $y \in [0,1]$, $F(y) - F(y/2) \le \gamma$. 

    We note one property that follows from the distribution $f$. For all candidates $k \ge 2$, 
    
    \begin{align*}
        U_k^L(X) & = F(x_k) - F((x_k + x_{k-1})/2) \\
        & \le F(x_k) - F(x_k/2) \le \gamma \, .
    \end{align*}
    
    We next show that for candidates $k \le m-1$, $U_k^R(X) \le \epsilon + \gamma + M \delta$. These utilities are depicted in Figure~\ref{fig:general-lb}.
    
    To see this, suppose for a contradiction that for some $k \le m-1$, $U_k^R(X) > \epsilon + \gamma + M \delta$. Then consider the profile $X'$ where candidate $k+1$ deviates to location $x_{k+1}' = x_k + \delta$. The increase in votes for candidate $k+1$ by deviating is

    \begin{align}
    U_{k+1}(X') - U_{k+1}(X) & \ge U_{k+1}^R(X') - U_{k+1}(X) \label{eqn:m-candidates} \, .
    \end{align}

    We note that before candidate $k+1$'s deviation, the total mass of voters in the interval $[x_k,x_{k+2}]$ is $U_k^R(X) + U_{k+1}(X) + U_{k+2}^L(X)$ (see Figure~\ref{fig:general-lb}). After the deviation, this mass is at most $M \delta + U_{k+1}^R(X') + U_{k+2}^L(X')$, hence 

    \begin{align*}
        U_{k+1}^R(X') - U_{k+1}(X) & \ge \left( U_k^R(X) + U_{k+2}^L(X) \right) -\\ &\hspace{0.6cm} \left( M \delta +U_{k+2}^L(X')\right) \, .
    \end{align*}

    Substituting in~\eqref{eqn:m-candidates}, and neglecting the term $U_{k+2}^L(X)$, we obtain

    \begin{align*}
    U_{k+1}(X') - U_{k+1}(X) & \ge U_k^R(X) - M \delta - U_{k+2}^L(X') \, .
    \end{align*}

\begin{figure}[H]
    \centering
    \begin{tikzpicture}[line cap=round,line join=round,>=triangle 45,x=1cm,y=1cm,scale=0.4]
\clip(-1,-3) rectangle (21,9.35);
\draw [line width=1pt] (0,0) -- (0,9.35);
\draw [dashed,line width=1pt] (14,0) -- (14,2.5);
\draw [dashed,line width=1pt] (4,0) -- (4,6.9);
\draw [line width=1pt] (20,0) -- (20,1);
\draw [line width=1pt] (8,0) -- (8,4.85);
\draw [line width=1pt] (0,0) -- (20,0);
\draw[color=black, font =\small ] (1.95,4) node {$\leq \epsilon+\gamma$};
\draw[color=black, font =\small ] (2.3,3) node {$ + M\delta$};
\draw[color=black, font = \small] (11,1.5) node {$\leq \epsilon+\gamma + M\delta$};
\draw[color=black, font = \small] (6,3) node {$\leq \gamma $};
\draw[color=black, font = \small] (17,1) node {$\leq \gamma $};
\begin{small}
\draw [fill=xdxdff] (0,0) circle (8pt);
\draw[color=xdxdff] (0.16,-0.78) node {$x_{k}$};
\draw [fill=xdxdff] (20,0) circle (8pt);
\draw[color=xdxdff] (20.16,-0.78) node {$x_{k+2}$};
\draw [fill=xdxdff] (8,0) circle (8pt);
\draw[color=xdxdff] (8.16,-0.78) node {$x_{k+1}$};
\draw[color=black] (13.76,-0.78) node {$\frac{x_{k+1}+x_{k+2}}{2}$};
\draw[color=black] (4.24,-0.78) node {$\frac{x_k+x_{k+1}}{2}$};
\draw[color=red] (16,3) node {$f(z)$};

\end{small}
\draw [red,thick] plot [smooth,samples=200, tension=1] coordinates { 
   (0,9.35) (10,4) (20,1)};
\end{tikzpicture}
    \caption{Bounds on the utility of candidate $k+1$ in Theorem~\ref{thm:general-lb}.}
    \label{fig:general-lb}
\end{figure}
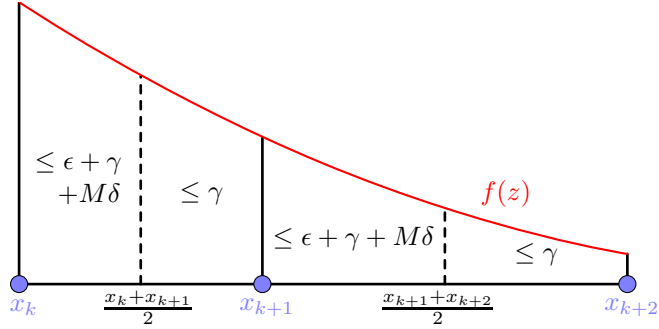

    For the expression on the right, $U_k^R(X) > \epsilon + \gamma + M \delta$ by assumption. Also, $U_{k+2}^L(X') = F(x_{k+2}) - F((x_{k+2} + x_{k+1}')/2)$ $\le F(x_{k+2}) - F(x_{k+2}/2)$ $\le \gamma$ by the properties of $f$. Hence, $U_{k+1}(X') - U_{k+1}(X) > \epsilon$, and $X$ is not an $\epsilon$-equilibrium. Hence for all candidates $k \le m-1$, $U_k^R(X) \le \epsilon + \gamma + M \delta$.

    Thus for all candidates $k \in \{2, \ldots, m-1\}$, the total utility is $U_k(X) = U_k^L(X) + U_k^R(X)$ $\le \epsilon + 2 \gamma + M \delta$. It follows then that for candidates $1$ and $k$, $U_1^L(X)$ and $U_k^R(X)$ are at most $\le 2\epsilon + 2 \gamma + 2 M \delta$.

    Finally, note that for any $k \le m-1$, $F(x_{k+1}) - F(x_k)$ $= U_k^R(X) + U_{k+1}^L(X)$ $\le \epsilon + 2\gamma + M \delta$. Thus,

    \begin{align*}
        1 & = F(x_1) + \sum_{k=1}^{m-1} \left(F(x_{k+1}) - F(x_k)\right) + \left(1-F(x_k)\right) \\
            & \le (2 (\epsilon + \gamma + M \delta) + (m-1) (\epsilon + 2 \gamma + 2 M \delta) + \\
            &\hspace{5.5cm}2 (\epsilon + \gamma + M \delta)) \\
            & = (m+3) \epsilon + 2(m+1) (\gamma + M \delta)
    \end{align*}

    \noindent and hence, as $\gamma, \delta \rightarrow 0$, $\epsilon \ge 1/(m+3)$. This completes the proof.
    \end{proof}

\section{Approximate Equilibria for 3 Candidates in a Variant}

Our previous results require that two candidates can be arbitrarily close, but cannot occupy the same location. In effect, even for candidates that are nearly identical, voters are able to perfectly differentiate between the two. It is natural to wonder if our results hold in the limit. That is, if we relax this assumption, and allow multiple candidates to occupy the same position (and thus voters fail to differentiate between these candidates), do our results hold? While it may appear that upper bounds on approximate equilibria in the original model should hold for this model as well, due to the larger space of strategies allowed, we do not know if this is indeed the case. In this section, we however give some evidence that this is indeed true. In particular, for the case of 3 candidates, we show that for any distribution, we can obtain a $1/7$-equilibrium, and this is tight: there exist distributions for which we cannot do better than a $1/7$-equilibrium.

A slight redefinition of utilities is required in this model. We specify that if there are $r$ candidates at a given location $x$, and $x$ is closer to $\gamma$ voters than any other candidate, then each candidate at $x$ gets $\gamma/r$ votes. Formally, given a strategy profile $X = (x_1, \ldots, x_k)$, where $x_1 \le \ldots \le x_k$, say $x_{i-1} < x_i = \ldots = x_j < x_{j+1}$. Then candidates $i, \ldots, j$ get utility 

\[
\frac{F\left((x_{j+1} + x_j)/2\right) - F\left((x_i + x_{i-1})/2\right)}{j-i+1} \, . 
\]
The definition of $\epsilon$-equilibrium is also modified appropriately.
\begin{definition}[$\epsilon$-equilibrium]
    Given an $\epsilon \ge 0$, $X = (x_1,x_2,...,x_m)$ is an $\epsilon$-equilibrium if for any candidate $i\in[m]$, and any location $x_i'\in[0,1]$,
    \[
        U_i(X')-U_i(X)\le\epsilon
    \]
    where $X'=(x_i',X_{-i})$.
\end{definition}

\begin{theorem}\label{thm:1/7ub}
    For 3 candidates, if we allow multiple candidates to choose the same locations, then given a distribution $f$ of voters, there always exists a $\frac{1}{7}$-equilibrium.
\end{theorem}

\begin{proof}
    For the proof, we define $\delta > 0$ as a small constant, so that $\delta M<10^{-3}$. We will actually show a $\frac{1}{7} + M \delta$-equilibrium, and then take $\delta\to0$. We define three key points on the interval: $z_{2/7} = \text{Cut}(0,2/7)$, $z_{5/7} = \text{Cut}(0,5/7)$, and $z_{2,\max}$ is the point between $(z_{2/7}+\delta)$ and $(z_{5/7}-\delta)$ with maximum votes. Thus $z_{2,\max} = \arg \max_{z \in [z_{2/7}+\delta, z_{5/7}-\delta]} U_2((z_{2/7}, z, z_{5/7}))$. The strategy profile $X^0 = (z_{2/7}, z_{2,\max}, z_{5/7})$ is our initial strategy profile, which we will modify to get a $1/7$-equilibrium. 
    
    We first show that in this initial strategy profile, candidate 2 gets at least $1/7$ votes.
    
    \begin{clm}
        $U_2(X^0) \ge 1/7$ \, .
    \end{clm}
    \begin{proof}
        Note that $F(z_{5/7}) - F(z_{2/7}) = 3/7$, hence, placing candidate 2 either just after $z_{2/7}$ or just before $z_{5/7}$ should give him at least half of these votes. That is, either $F((z_{2/7} + z_{5/7})/2) - F(z_{2/7}) \ge 1.5/7$ or $F(z_{5/7}) - F((z_{2/7} + z_{5/7})/2) \ge 1.5/7$. In either case, since $z_{2,\max}$ is the position in $[z_{2/7}+\delta, z_{5/7}-\delta]$ where candidate 2 gets maximum votes, he must get at least $1/7$ votes at $z_{2,\max}$.
    \end{proof}
    
    We will consider two cases: the first, where $U_2(X^0) \le 2/7$, and the second, where $U_2(X^0) > 2/7$.

    \paragraph*{Case I: $U_2(X^0) \le 2/7$.} Here, we consider two subcases. For the easy case, suppose both $U_2^L(X^0) \le 1/7$ and $U_2^R(X^0) \le 1/7$. Then we claim that in fact $X^0 = (z_{2/7}, z_{2,\max}, z_{5/7})$ is a $1/7$-equilibrium. 
    \begin{enumerate}
        \item For candidate 2, note that he cannot improve his votes in the interval $(z_{2/7},z_{5/7})$ by more than $\delta M$. At any position before $z_{2/7}$ or after $z_{5/7}$, he would get at most $2/7$ votes, and could thus improve by at most $1/7$. At the positions $z_{2/7}$ and $z_{5/7}$, he would get at most $(U_1^R(X^0)+2/7)/2\le$ $(U_2(X^0)+\delta M + 2/7)/2$ $< 1/7 + U_2(X^0)$.
        \item For candidates 1 and 3, since $U_2^L(X^0) \le 1/7$ and $U_2^R(X^0) \le 1/7$, they can at best increase their votes by $1/7$.
    \end{enumerate}

    For the second subcase, either $U_2^L(X^0) > 1/7$ or $U_2^R(X^0) > 1/7$. Clearly, both cannot hold, since we assume  $U_2(X^0) \le 2/7$. Without loss of generality, assume $U_2^R(X^0) > 1/7$ and $U_2^R(X^0) \le 1/7$. Then by continuity of the function $F$, there must exist some $\eta > 0$ so that at the strategy profile $X' = (z_{2/7}, z_{2,\max} + \eta, z_{5/7})$, candidate 2 gets exactly $1/7$ votes from the right. That is, $U_2^R(X') = 1/7$. Further, $U_2^L(X') \le 1/7$ (because he receives a maximum of $2/7$ votes at $z_{2,\max}$), and $U_1^L(X') = U_3^R(X') = 2/7$. We claim that in this case, $X'$ is a $1/7$-equilibrium.
    \begin{enumerate}
        \item For candidate 2, in the interval $(z_{2/7},z_{5/7})$ he gets at most $2/7 + M \delta$ votes (since at $z_{2,\max}$ he gets at most $2/7$ votes), and hence cannot improve by more than $1/7+\delta M$. Similarly, before $z_{2/7}$ and after $z_{5/7}$ he gets at most $2/7$ votes. At the positions $z_{2/7}$ and $z_{5/7}$, he would get at most $(U_2(X^0) + 2/7 +\delta M)/2$ $\le 1/7 + U_2(X^0)/2 + \delta M/2$ $\le 2/7 + \delta M/2$.
        \item For candidates 1 and 3, as previously, since $U_2^L(X') \le 1/7$ and $U_2^R(X') \le 1/7$, they can at best increase their votes by $1/7$.
    \end{enumerate}

    Thus, if $U_2(X^0) \le 2/7$, either $X^0$ or $X'$ is a $(1/7 + M \delta)$-equilibrium.

    \paragraph*{Case II: $U_2(X^0)>2/7$.} Recall that $X^0 = (z_{2/7}, z_{2,\max}, z_{5/7})$. We will consider three subcases here. For the first subcase, assume that for some location $x_2' \in [z_{2/7}+ \delta, z_{5/7} - \delta]$, in the strategy profile $X' = (z_{2/7}, x_2', z_{5/7})$, candidate 2 gets at least $1/7$ votes from both the left and the right. That is, both $U_2^L(X') \ge 1/7$ and $U_2^R(X') \ge 1/7$. Then there exist $\eta_1 \ge 0$, $\eta_2 \ge 0$ so that for the strategy profile $X'' = (z_{2/7} + \eta_1, z_{2,\max}, z_{5/7} - \eta_2)$, candidate 2 gets exactly $1/7$ from each side. That is, $U_2^L(X'') = U_2^R(X'') = 1/7$. We claim that $X''$ is a $1/7$-equilibrium. Note that each candidate in this profile gets at least $2/7$ votes. Further, each interval $[0, z_{2/7} + \eta_1]$, $[z_{2/7} + \eta_1, z_{5/7} - \eta_2]$, and $[z_{5/7} - \eta_2, 1]$ contains at least $2/7$ votes, and hence each of these intervals has at most $3/7$ votes. Figure~\ref{fig:variant-1} shows the positions of the candidates and the voters in each interval.

\begin{figure}[H]
    \centering
    \begin{tikzpicture}[line cap=round,line join=round,>=triangle 45,x=1cm,y=1cm,scale=0.382]
    \clip(-1,-3) rectangle (22,4);
\draw [line width=1pt, color = red] (3,-0.5) -- (3,0.5);
\draw [dashed,line width=1pt] (13.5,-1) -- (13.5,1);
\draw [dashed,line width=1pt] (7.5,-1) -- (7.5,1);
\draw [line width=1pt] (21,-0.3) -- (21,0.3);
\draw [line width=1pt] (0,-0.3) -- (0,0.3);
\draw [line width=1pt, color = red] (19,-0.5) -- (19,0.5);
\draw [line width=1pt] (0,0) -- (21,0);
\draw [line width=0.3pt,<-] (0,2) -- (1.6,2);
\draw [line width=0.3pt,->] (3.5,2) -- (5,2);
\draw [line width=0.3pt,<-] (16.7,2) -- (18,2);
\draw [line width=0.3pt,->] (19.8,2) -- (21,2);
\draw [line width=0.3pt,<->] (7.5,0.8) -- (10,0.8);
\draw [line width=0.3pt,<->] (10.1,0.8) -- (13.4,0.8);
\begin{small}
\draw [fill=xdxdff] (5,0) circle (8pt);
\draw[color=red] (3.16,-1) node {$z_{2/7}$};
\draw [fill=xdxdff] (17,0) circle (8pt);
\draw[color=red] (19.16,-1) node {$z_{5/7}$};
\draw [fill=xdxdff] (10,0) circle (8pt);
\draw[color=xdxdff] (10.16,-0.8) node {$x_2'$};
\draw[color=black] (0,-0.8) node {$0$};
\draw[color=black] (21,-0.8) node {$1$};
\draw[color=black] (2.55,2) node {$\leq \frac{3}{7}$};
\draw[color=black] (8.8,1.6) node {$\frac{1}{7}$};
\draw[color=black] (12,1.6) node {$\frac{1}{7}$};
\draw[color=black] (19,2) node {$\le \frac{3}{7}$};
\end{small}
\end{tikzpicture}
    \caption{Figure for the case where for some location $x_2' \in (z_{2/7},z_{5/7})$, the candidate gets at least $1/7$ votes on either side. The blue circles denote the candidate locations in the strategy profile $X''$, and the dashed lines denote boundaries.}
    \label{fig:variant-1}
\end{figure}
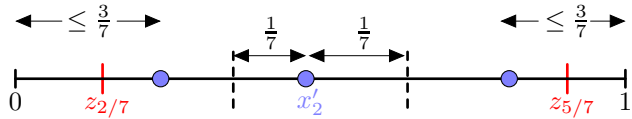
    
    \begin{enumerate}
        \item For candidate 2, since he already gets exactly $2/7$ votes and each interval above contains at most $3/7$ votes, he cannot increase his votes by more than $1/7$.
        \item For candidates 1 and 3, they already receive at least $2/7$ votes. Since $U_2^L(X'') = U_2^R(X'') = 1/7$, and each interval above contains at most $3/7$ votes, they can increase their votes by at most $1/7$.
    \end{enumerate}

    This concludes the first subcase. For the second and third subcases, since $U_2(X^0) > 2/7$, we assume without loss of generality that $U_2^L(X^0) < 1/7$ and $U_2^R(X^0) \ge 1/7$. We define the point $z_{3/7} = \text{Cut}(0,3/7)$ that will be useful. We consider two subcases: (i) $z_{3/7} \le z_{2, \max} $, and (ii) $z_{3/7} > z_{2,\max}$, and show how to get a $1/7$-equilibrium in each case.

    For the first subcase, we thus have $z_{3/7} \le z_{2, \max} $, and also that $U_2^L(X^0) < 1/7$, $U_2^R(X^0) \ge 1/7$, and $U_2(X^0) > 2/7$. Then since $z_{3/7} \le z_{2, \max} $ and $U_2^R(X^0) \ge 1/7$, let $x_2'$ be the first point in $[z_{3/7}, z_{2, \max}]$ where candidate 2 gets at least $1/7$ votes from the right, i.e., $F((x_2' + z_{5/7})/2) - F(x_2') \ge 1/7$.

    Let's first assume that $x_2' = z_{3/7}$. Then  $F((z_{3/7} + z_{5/7})/2) - F(z_{3/7}) \ge 1/7$. There exists $\eta \ge 0$ so that $F((z_{3/7} + z_{5/7} - \eta)/2) - F(z_{3/7}) = 1/7$. Then consider the strategy profile $X' = (z_{3/7}, z_{3/7}, z_{5/7} - \eta)$. Here, $U_1(X') = U_2(X') = 1/2 \times 4/7 = 2/7$, and also $U_3(X') \ge 2/7$. It can be easily checked that $X'$ is a $1/7$-equilibrium. 

    If $x_2' > z_{3/7}$, then clearly $F((z_{3/7} + z_{5/7})/2) - F(z_{3/7}) < 1/7$. By continuity of $F$, it must be the case that $F((x_2' + z_{5/7})/2) - F(x_2') = 1/7$. Then consider the strategy profile $X' = (z_{2/7}, x_2', z_{5/7})$. Note that for candidate 2, $U_2^R(X') = 1/7$, and hence $U_2^L(X') < 1/7$, else we would be in the first subcase of Case II. Further, both candidates 1 and 3 get at least $2/7$ votes each. Clearly, candidates 1 and 3 cannot increase their votes by more than $1/7$. We will now show that candidate 2 also cannot increase his votes by more than $1/7+\delta M$. We first show that at $x_2'$, candidate 2 loses at most $1/7$ votes compared to the location $z_{2,\max}$ which immediately tells us that candidate 2 cannot increase his votes by more than $1/7+\delta M$. The candidate positions and voters for this case are depicted in Figure~\ref{fig:variant2}.

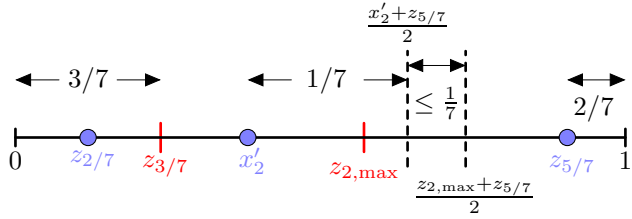
\begin{figure}[H]
    \centering
    \begin{tikzpicture}[line cap=round,line join=round,>=triangle 45,x=1cm,y=1cm,scale=0.382]
    \clip(-1,-3) rectangle (22,5);
\draw [line width=1pt, color = red] (5,-0.5) -- (5,0.5);
\draw [dashed,line width=1pt] (13.5,-1) -- (13.5,3);
\draw [dashed,line width=1pt] (15.5,-1) -- (15.5,3);
\draw [line width=1pt] (21,-0.3) -- (21,0.3);
\draw [line width=1pt] (0,-0.3) -- (0,0.3);
\draw [line width=1pt, color = red] (12,-0.5) -- (12,0.5);
\draw [line width=1pt] (0,0) -- (21,0);
\draw [line width=0.3pt,<-] (0,2) -- (1.5,2);
\draw [line width=0.3pt,->] (3.5,2) -- (5,2);
\draw [line width=0.3pt,<->] (19,2) -- (21,2);
\draw [line width=0.3pt,<->] (13.5,2.5) -- (15.5,2.5);
\draw [line width=0.3pt,<-] (8,2) -- (9.5,2);
\draw [line width=0.3pt,->] (12,2) -- (13.5,2);
\begin{small}
\draw [fill=xdxdff] (2.5,0) circle (8pt);
\draw[color=xdxdff] (2.66,-0.8) node {$z_{2/7}$};
\draw[color=red] (5.16,-1) node {$z_{3/7}$};
\draw [fill=xdxdff] (19,0) circle (8pt);
\draw[color=xdxdff] (19.16,-1) node {$z_{5/7}$};
\draw [fill=xdxdff] (8,0) circle (8pt);
\draw[color=xdxdff] (8.16,-0.8) node {$x_2'$};
\draw[color=black] (0,-0.8) node {$0$};
\draw[color=black] (21,-0.8) node {$1$};
\draw[color=black] (2.5,2) node {$3/7$};
\draw[color=black] (10.7,2) node {$1/7$};
\draw[color=black] (14.5,1.2) node {$\leq\frac{1}{7}$};
\draw[color=red] (12,-1.2) node {$z_{2,\max}$};
\draw[color=black] (19.9,1) node {$2/7$};
\draw[color=black] (13.5,4) node {$\frac{x_2' + z_{5/7}}{2}$};
\draw[color=black] (15.8,-2) node {$\frac{z_{2,\max} + z_{5/7}}{2}$};
\end{small}
\end{tikzpicture}
    \caption{Figure for Claim~\ref{clm:u2}, depicting the candidate locations, and change in candidate $2$'s utility as he moves from $z_{2,\max}$ to $x_2'$. }
    \label{fig:variant2}
\end{figure}

    \begin{clm}
        $U_2(X^0) - U_2(X') \le 1/7$.
        \label{clm:u2}
    \end{clm}
    \begin{proof}
        Recall $X^0 = (z_{2/7}, z_{2,\max}, z_{5/7})$, $X' = (z_{2/7}, x_2', z_{5/7})$, and $x_2' \le z_{2,\max}$. Hence in shifting to $x_2'$, candidate 2 gains votes on the left and loses votes on the right. Thus:

        \begin{align*}
        U_2(X^0) - U_2(X') &\le (F((z_{2,\max} + z_{5/7})/2) - \\
        &\hspace{3cm}F((x_2' + z_{5/7})/2)) \, .
        \end{align*}

        \noindent Since $z_{2,\max} \le z_{5/7}$,

          \[
        U_2(X^0) - U_2(X') \le F(z_{5/7}) - F((x_2' + z_{5/7})/2)\, .
        \]

        We now show that the expression on the left is at most $1/7$, to complete the proof. Note that $x_2' \ge z_{3/7}$, and $U_2^R(X') = F((x_2' + z_{5/7})/2) - F(x_2')$ $= 1/7$, and hence $F((x_2' + z_{5/7})/2) \ge 4/7$. It follows by definition of $z_{5/7}$ that $F(z_{5/7}) - F((x_2' + z_{5/7})/2) \le 1/7$, as required.
    \end{proof}
    
    For the last subcase, we thus have $z_{3/7} > z_{2, \max} $, and also that $U_2^L(X^0) < 1/7$, $U_2^R(X^0) \ge 1/7$, and $U_2(X^0) > 2/7$. Then since $z_{3/7} > z_{2, \max} $ and $U_2^R(X^0) \ge 1/7$, let $x_2'$ be the last point in $[z_{2, \max}, z_{3/7}]$ where candidate 2 gets at least $1/7$ votes from the right, i.e., $F((x_2' + z_{5/7})/2) - F(x_2') \ge 1/7$. This subcase is then similar to the previous subcase, where we consider either $x_2' = z_{3/7}$, or $x_2' < z_{3/7}$, and a $1/7$-equilibrium can be found accordingly.
\end{proof}

Note that the proof of Theorem \ref{thm:1/7ub} only shows the existence of a $1/7$-approximate equilibrium. However, to obtain a polynomial time algorithm that achieves a $1/7$-approximate equilibrium, we require a more robust oracle. This oracle should perform the following tasks: given an interval $[x,y]$, it returns a sub-interval $[x',y'] \subset [x,y]$ such that (i) $y' - x' = (y-x)/2$, and (ii) $F(y')-F(x')$ is maximized among all such sub-intervals. The location $z_{2,\max}$ as defined in the proof of Theorem \ref{thm:1/7ub} corresponds to the midpoint $(x'+y')/2$ of such an interval. Unfortunately, it remains unclear how to identify such a sub-interval in polynomial time using only Cut and Eval queries. 

The subsequent theorem asserts the tightness of this $1/7$-approximate equilibrium.

\begin{theorem}\label{thm:lb31}
    For 3 candidates, if we allow multiple candidates to choose the same location, then there exists a distribution $f$ of voters such that for any $\epsilon$-equilibrium, $\epsilon\ge\frac{1}{7}$.
\end{theorem}
    
\begin{proof}
    Fix $\gamma>0$. We choose the distribution $f$ of the voters as per Theorem \ref{magicfunction}, so that $f$ is bounded, integrable, $\int_0^1 f(z) dz = 1$, and for any $y \in [0,1]$, $F(y) - F(y/2) \le \gamma$. Suppose $X=(x_1,x_2,x_3)$ is an $\epsilon$-equilibrium. We divide our proof into three cases according to the maximum number of candidates at the same location.\\\\
    \textbf{Case I:} Suppose $x_1<x_2<x_3$. In this case, from Theorem \ref{thm:general-lb}, we conclude that $\epsilon\ge1/6$.\\\\
    \textbf{Case II:} Suppose $x_1=x_2=x_3$. Each candidate then receives $1/3$ votes. Any candidate can get at least $1/2$ votes by deviating from that location, and hence $\epsilon\ge1/6$.\\\\

\begin{figure}[!htp]
    \centering
    \begin{tikzpicture}[line cap=round,line join=round,>=triangle 45,x=1cm,y=1cm,scale=0.4]
\clip(-1,-3) rectangle (21,14);
\draw [line width=1pt] (0,0) -- (0,14);
\draw [dashed,line width=1pt] (8,0) -- (8,5.3);
\draw [color = blue, dashed, line width=1pt] (2.2,0) -- (2.2,11.5);
\draw [line width=1pt] (10,0) -- (10,4);
\draw [line width=1pt] (6,0) -- (6,7);
\draw [line width=1pt] (0,0) -- (20,0);

\draw[color=black, font = \normalsize] (7,3) node {$\leq \epsilon $};
\draw[color=black, font = \normalsize] (9,2) node {$\leq \gamma $};
\draw[color=black, font = \normalsize] (13,1) node {$\leq 3\epsilon $};
\draw[color=black, font = \normalsize] (4,4) node {$\leq 3\epsilon $};
\begin{small}
\draw [fill=xdxdff] (6,0) circle (8pt);
\draw[color=xdxdff] (5.9,-0.78) node {$x_{1}=x_2$};
\draw [fill=xdxdff] (10,0) circle (8pt);
\draw[color=xdxdff] (10,-0.78) node {$x_{3}$};
\draw[color=black] (2.2,-0.78) node {$e^{-\frac{1}{\gamma}}$};
\draw[color=black] (0,-0.78) node {$0$};
\draw[color=red] (16,3) node {$f(z)$};

\end{small}
\draw [red,thick] plot [smooth,samples=200, tension=1] coordinates { 
   (2.2,11.8) (10,4) (20,1)};
\end{tikzpicture}
    \caption{Bounds on the utilities of candidates 1, 2, and 3 in Case III of Theorem \ref{thm:lb31}, with the midpoint of $x_1$ and $x_3$ indicated by the black dashed line.}
    \label{fig:lb31}
\end{figure}
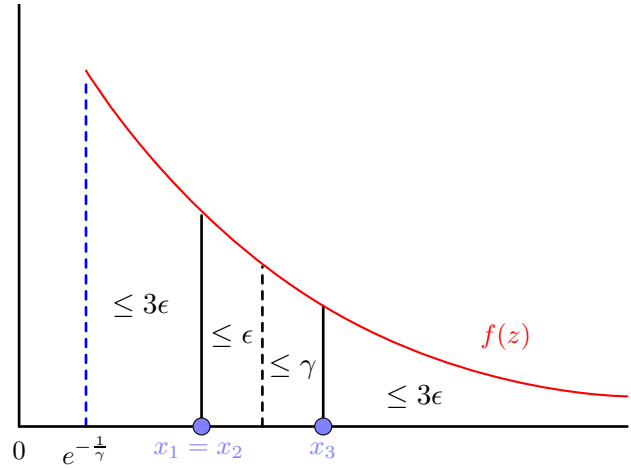
    
  \noindent  \textbf{Case III:} Suppose $x_1=x_2<x_3$. The votes in this case are depicted in Figure~\ref{fig:lb31}. Note that $F((x_1+x_3)/2)-F(x_1)\le\epsilon$, otherwise candidate 3 can shift arbitrarily close to the right of $x_1$ and can increase his utility by more than $\epsilon$. Further we have $F(x_3)-F((x_1+x_3)/2)\le\gamma$ (by definition of $f$). Therefore, $F(x_3)-F(x_1)\le \epsilon+\gamma$.

    \begin{clm}\label{clm:lb31}
        $U_1(X)\le(F(x_1)+\epsilon)/2$.
    \end{clm}
    \begin{proof}
        We bound the utility of the first candidate as follows.
        \begin{align*}
            U_1(X)&=F((x_1+x_3)/2)/2 \\
            &= (F(x_1)+F((x_1+x_3)/2)-F(x_1))/2\\
            &\le (F(x_1)+\epsilon)/2
        \end{align*}
    \end{proof}

    \begin{corollary}\label{cly:lb31}
        $F(x_1)\le 3\epsilon$.
    \end{corollary}
    \begin{proof}
        For the sake of contradiction suppose $F(x_1)=3\epsilon+\eta$ for some $\eta>0$. From Claim \ref{clm:lb31}, we conclude that $U_1(X)\le 2\epsilon + \eta/2$. Now there exists some $x_1'<x_1$ such that $F(x_1')>3\epsilon+\eta/2$. If candidate 1 deviates to $x_1'$, his utility will increase by more than $3\epsilon+\eta/2-(2\epsilon +\eta/2)=\epsilon$. This gives a contradiction.
    \end{proof}
The next two bounds are immediate.
    
    \begin{corollary}\label{cly:lb32}
        $U_1(X)\le 2\epsilon$.
    \end{corollary}
    \begin{corollary}\label{cly:lb33}
        $1-F(x_3)\le 3\epsilon$.
    \end{corollary}
    Now we can show an lower bound on $\epsilon$ as follows.
    \begin{align*}
        1 &= \int_0^1 f(z) dz\\
        &= F(x_1) + (F(x_3)-F(x_1)) + (1 - F(x_3))\\
        &\le 3\epsilon+\epsilon+\gamma+3\epsilon ~~~[\text{From Corollary~\ref{cly:lb31} and Corollary~\ref{cly:lb33}}]\\
        &=7\epsilon+\gamma
    \end{align*}
    Therefore, for $\gamma\rightarrow0,~\epsilon\ge\frac{1}{7}$.
\\\\
\textbf{Case IV:} Suppose $x_1<x_2=x_3$. Note that $F(x_2)-F((x_1+x_2)/2)\le\gamma$ (by definition of $f$). Therefore, $U_2(X)=(1-F((x_1+x_2)/2))/2\le(1-F(x_2)+\gamma)/2$. Similar to Corollaries~\ref{cly:lb31}, ~\ref{cly:lb32} and ~\ref{cly:lb33} of Case III, we can show that $1-F(x_2)\le (2\epsilon+\gamma)$, $F(x_1)\le(2\epsilon+\gamma)$, and $U_2(X)\le\epsilon+\gamma$. Notice that $F((x_1+x_2)/2)-F(x_1)\le(2\epsilon+\gamma)$ otherwise candidate 2 can shift arbitrarily close to the right of $x_1$ and can increase his utility by more than $\epsilon$. The proof of the theorem can be concluded as follows.
\begin{align*}
        1 &= \int_0^1 f(z) dz\\
        &= F(x_1) + (F(x_2)-F(x_1)) + (1 - F(x_2))\\
        &\le (2\epsilon+\gamma)+(2\epsilon+2\gamma)+(2\epsilon+\gamma) ~~~[\text{From Corollary \ref{cly:lb31} and \ref{cly:lb33}}]\\
        &=6\epsilon+4\gamma
    \end{align*}
    Therefore, for $\gamma\rightarrow0,~\epsilon\ge\frac{1}{6}$.
\end{proof}

\section{Conclusion}

In the basic model of spatial competition in the unit interval, our work gives nearly tight best-case and worst-case bounds on how close we can get to equilibria, quantifying the instability due to lack of equilibria. As in prior work, it would be interesting to see if these bounds are robust to slight changes in the model, such as if we allow multiple candidates to occupy the same location. Another assumption that we would like to relax is that of inelastic demand, which requires that each voter must vote, even if the nearest candidate is very far from his location. Given the emphasis placed in elections on ``turning out the vote,'' it seems to us that voter abstention or apathy is a crucial aspect that should be captured. More generally, we believe that quantifying instability in this manner may be a useful and interesting line of research. 

In our model, we assume that the utilities of the candidates are normalised, and we study additive approximations to equilibria. We believe this is well-motivated. For normalised utilities, additive approximations to equilibria are more robust than multiplicative approximations. For example, in our model, consider two candidate strategy profiles $X$, $X'$, so that for any candidate $i$, $|x_i-x_i'| \le \delta$. Then if $X$ is an additive $\epsilon$-approximate equilibrium, $X'$ is an additive $(\epsilon + 4M \delta)$-equilibrium, where $M \ge f(z)$ for all $z \in [0,1]$. However if $X$ is a multiplicative $\epsilon$-approximate equilibrium, we cannot give any similar guarantees for $X'$. Further, lower bounds on additive approximations directly imply lower bounds on multiplicative approximations. However obtaining multiplicative bounds on approximate equilibria is clearly an important question, which remains open for future work.

\bibliographystyle{plain}
\bibliography{hotelling-approxeq-arxiv}
\end{document}